\pdfoutput=1
\documentclass[12pt]{article}

\usepackage[T1]{fontenc}
\usepackage{lmodern}
\usepackage[a4paper,margin=1in]{geometry}
\usepackage{amsmath,amssymb,amsthm}
\usepackage{booktabs}
\usepackage{array}
\usepackage{graphicx}
\usepackage{comment}
\usepackage{microtype}
\usepackage[round,authoryear]{natbib}
\usepackage[colorlinks=true,allcolors=blue,hypertexnames=false]{hyperref}

\newtheorem{lemma}{Lemma}
\newtheorem{proposition}{Proposition}
\newtheorem{corollary}{Corollary}

\newcommand{\R}{\mathbb{R}}
\newcommand{\ind}[1]{\mathbf{1}\{#1\}}
\newcommand{\Norm}{\mathcal{N}}
\DeclareMathOperator{\Phinorm}{\Phi}
\begin{document}
\title{Defense Feedback and Global-Game Selection}
\author{Georgy Lukyanov\thanks{Corresponding author: \href{mailto:georgy.lukyanov@tse-fr.eu}{georgy.lukyanov@tse-fr.eu}.}\\ \small Toulouse School of Economics, University of Toulouse Capitole\\ \small 1 Esplanade de l'Universit\'e, 31080 Toulouse Cedex 06, France \and Yaroslav Yosifov\\ \small International College of Economics and Finance\\ \small National Research University Higher School of Economics\\ \small 11 Pokrovsky Boulevard, building T, Moscow 109028, Russia}
\date{}
\maketitle

\begin{abstract}
Global games select a unique cutoff when a player's incentive moves monotonically with private information. We show that this logic can fail when a privately informed participant can take a costly, simultaneous, and unobserved action that defends the status quo. In a currency-attack model, the defender acts only at intermediate signals. Along the traders' cutoff equation, this interval response creates an endogenous defense-loss term that vanishes at the endpoints of an activation component but can dominate the decreasing benchmark payoff margin in between. We establish a primitive region in which defense is inactive and uniqueness survives, and a finite-CDF overshoot condition under which two active-defense cutoff roots coexist with the inactive benchmark root. Because marginal indifference need not imply global incentives, we also give sufficient conditions for a genuine monotone cutoff best response. Finally, outward-rounded interval arithmetic proves that a proper-prior parameter vector has exactly three symmetric monotone-cutoff equilibria and that this topology persists locally. The result isolates a payoff-based boundary of global-game selection: hidden intervention can restore multiplicity without public signalling or a perturbation of the information hierarchy.
\end{abstract}



\section{Introduction}

Private signals select a unique cutoff in the benchmark currency-attack global game because the marginal trader's payoff moves monotonically with her signal \citep{MorrisShin1998}.  We identify a payoff-based boundary of this selection argument.  One privately informed participant can take a costly binary action that subtracts capacity from aggregate attack pressure.  The action is simultaneous and unobserved: it changes the small traders' success probability but conveys no public information.

The crucial object is not the defender's size by itself, but the shape of his best response.  Defense is worthwhile only when fundamentals are neither so weak that finite capacity cannot save the peg nor so strong that protection is unnecessary.  The defender therefore acts on an intermediate-signal interval. Along the small traders' cutoff equation, the success probability removed by intervention vanishes at both ends of an activation component and can be large inside it.  This defense-loss term can cross the decreasing no-defense payoff margin twice before the usual inactive root is reached.

We first characterize when this feedback is impossible.  The defender's pivotal probability is always strictly below
\[
 \overline H=2\Phinorm\!\left(\frac{\sqrt{\alpha+\tau}\,\lambda}{2}\right)-1.
\]
If the defender's cost is at least \(\overline H\), he never acts and the standard cutoff equilibrium remains unique.  For lower costs, the residual decomposes as
\[
 F_\tau(b)=R_\tau(b)-L_\tau(b),
\]
where \(R_\tau\) is the strictly decreasing no-defense margin and \(L_\tau\) is the probability removed by intervention.  A finite sum of univariate Gaussian CDFs bounds \(L_\tau\) from below.  If that bound overtakes \(R_\tau\) within the left activation component, two active-defense roots coexist with the inactive benchmark root.  The construction requires an intermediate defense cost: defense must operate at pessimistic cutoffs but be inactive at the benchmark root.

Marginal indifference alone is not sufficient.  Because the probability of no defense is U-shaped in fundamentals, a trader's success probability need not be monotone in her signal.  We give conditions under which the downward jumps in state-contingent success dominate this interior variation and hence restore the global cutoff incentives.  At one proper-prior parameter vector, outward-rounded interval arithmetic then proves exactly three symmetric monotone-cutoff equilibria---two with active defense and one without---and excludes every additional cutoff root.  Strict inequalities establish local persistence while holding the prior fixed.

\citet{CorsettiEtAl2004} provide the closest large-participant benchmark.  In their simultaneous currency-attack game, a privately informed large trader's action is also hidden from the small traders, but the large trader joins the attack and uses a lower-signal cutoff; equilibrium is unique and dominance solvable under an improper uniform prior interpreted as a diffuse-prior limit. Here the large participant instead values peg survival.  His costly action is effective only in pivotal states, generating an interval response and the nonmonotone defense-loss term.  Corsetti et al.'s sequential extension concerns public action and signalling; their opposite-exposure discussion does not add costly defense capacity or derive multiplicity.

The result is distinct from multiplicity generated by observed policy actions and their informational content \citep{AngeletosEtAl2006,Lukyanov2025}. \citet{JannSchottmuller2021} study finitely many attackers and a costly, continuous, secret defense investment; their equilibrium is unique with sufficiently many potential attackers.  \citet{Mitkov2025} uses private, payoff-irrelevant sunspots to pin down a unique coordination probability without intrinsic fundamental uncertainty.  Our mechanism instead combines payoff-relevant global-game information, simultaneous hidden defense, and locally persistent cutoff multiplicity.

Nor is the claim that uniqueness is robust to arbitrary belief perturbations.  \citet{WeinsteinYildiz2007} show that refined predictions are generally sensitive to higher-order-belief perturbations.  We keep a proper Gaussian information structure fixed and locate a finite-dimensional, payoff-based boundary of selection.  The result is complete only within symmetric monotone small-trader cutoffs; asymmetric and non-cutoff equilibria remain outside its scope.

\section{Model}

There is a unit continuum of small traders \(i\in[0,1]\) and one large institutional investor, whom we call the defender.  Trader \(i\) chooses attack \(a_i\in\{0,1\}\); the defender chooses \(g\in\{0,1\}\) and has capacity \(\lambda\in(0,1)\).  Let \(M=\int_0^1a_i\,di\) be the attacking mass.  The peg collapses when
\[
 M-\lambda g\geq\theta,
\]
where \(\theta\) is the strength of fundamentals.  Normalized payoffs are
\[
 a_i\bigl(\ind{M-\lambda g\geq\theta}-c_S\bigr),
 \qquad
 \ind{M-\lambda g<\theta}-c_Lg,
\]
with \(c_L,c_S\in(0,1)\).  Interaction is simultaneous: the defender's choice is not observed before the small traders act. We identify strategies that differ only on zero-probability signal realizations.

The prior is \(\theta\sim\Norm(m_0,1/\tau)\), with \(\tau>0\).  Small trader \(i\) observes \(x_i=\theta+\varepsilon_i\), while the defender observes \(y=\theta+\eta\).  The shocks are independent, \(\varepsilon_i\sim\Norm(0,1/\beta)\) and \(\eta\sim\Norm(0,1/\alpha)\), with \(\alpha,\beta>0\).  We impose exact conditional aggregation: conditional on \(\theta\), the mass of traders with \(x_i\leq x\) equals \(\Phinorm(\sqrt\beta(x-\theta))\) almost surely.  This can be founded on a Fubini extension and the conditional exact law of large numbers \citep{Sun2006,QiaoSunZhang2016}.

We study symmetric monotone strategies in which small traders attack iff \(x_i\leq x^*\).  Their aggregate attack is then
\[
 M(\theta;x^*)=\Phinorm\!\left(\sqrt\beta(x^*-\theta)\right).
\]
Let \(a\) and \(b\) be the collapse thresholds with and without defense:
\begin{equation}\label{eq:thresholds}
 M(a;x^*)-\lambda=a,
 \qquad
 M(b;x^*)=b.
\end{equation}
Both are unique and \(a<b\).  Defense changes the outcome exactly in the pivotal states \(a<\theta\leq b\).

Let \(A=\alpha+\tau\), \(B=\beta+\tau\).  Posterior means are
\[
 m_L(y)=\frac{\alpha y+\tau m_0}{A},
 \qquad
 m_S(x)=\frac{\beta x+\tau m_0}{B},
\]
and posterior variances are \(1/A\) and \(1/B\), respectively.

\section{Benchmark uniqueness and the defense band}

Write \(\mu=(a+b)/2\) and \(h=(b-a)/2\).  Conditional on signal \(y\), the defender's gain from intervention is
\[
 \Pr(a<\theta\leq b\mid y)-c_L.
\]
The pivotal probability is symmetric and single-peaked in \(m_L(y)\), with maximum
\begin{equation}\label{eq:H}
 H_\tau(a,b)=2\Phinorm(\sqrt A h)-1.
\end{equation}
If \(H_\tau\leq c_L\), the defender never acts.  If \(H_\tau>c_L\), he defends on one signal interval \((y_L,y_H)=(\gamma-w,\gamma+w)\), where
\begin{equation}\label{eq:center}
 \gamma=\frac{A\mu-\tau m_0}{\alpha}
       =\mu+\frac{\tau(\mu-m_0)}{\alpha},
\end{equation}
and \(w>0\) is pinned down by his indifference at the two endpoints.

To parameterize a candidate cutoff, define
\[
 T_\beta(z)=z+\frac{\Phinorm^{-1}(z)}{\sqrt\beta},
 \qquad z\in(0,1).
\]
The function is strictly increasing from \((0,1)\) onto \(\R\).  Equation \eqref{eq:thresholds} is equivalent to
\begin{equation}\label{eq:parameterization}
 x^*(b)=T_\beta(b),
 \qquad
 T_\beta(p)=T_\beta(b)+\lambda,
 \qquad
 a=p-\lambda.
\end{equation}
Because \(p>b\), the pivotal-state width satisfies \(b-a<\lambda\).  As \(b\downarrow0\) or \(b\uparrow1\), \(p-b\to0\) and \(b-a\to\lambda\). It follows from \eqref{eq:H} that
\[
 H_\tau(a(b),b)<\overline H
 :=2\Phinorm\!\left(\frac{\sqrt A\lambda}{2}\right)-1
\]
for every \(b\in(0,1)\).

In the absence of defense, the marginal-trader residual is
\begin{equation}\label{eq:benchmark}
 R_\tau(b)=
 \Phinorm\!\left(\sqrt B\,[b-m_S(T_\beta(b))]\right)-c_S.
\end{equation}
If \(q=\Phinorm^{-1}(b)\), the normal index in \eqref{eq:benchmark} is
\[
 z(b)=\frac{\tau(b-m_0)-\sqrt\beta\,q}{\sqrt B},
 \qquad
 z'(b)=\frac{\tau-\sqrt\beta/\phi(q)}{\sqrt B}.
\]
Since \(\phi(q)\leq1/\sqrt{2\pi}\), the diffuse-prior condition
\begin{equation}\label{eq:diffuse}
 \tau<\sqrt{2\pi\beta}
\end{equation}
makes \(R_\tau\) strictly decreasing.  Its endpoint limits are \(1-c_S\) and \(-c_S\), so it has a unique root, denoted \(b_0\).  In the diffuse-prior limit \(\tau\downarrow0\), \(b_0=1-c_S\).

\begin{proposition}\label{prop:unique}
Suppose \eqref{eq:diffuse} holds.  If \(c_L\geq\overline H\), the defender is inactive at every monotone-cutoff candidate, and the game has a unique symmetric monotone-cutoff equilibrium, at \(b=b_0\).
\end{proposition}

\begin{proof}
The bound on \(H_\tau\) makes defense strictly unprofitable for every candidate.  The cutoff residual is therefore \(R_\tau\), which crosses zero once.  Without defense, a trader's success probability is a normal CDF that is strictly decreasing in her signal, so marginal indifference supplies the full cutoff incentives.
\end{proof}

\section{The defense-loss mechanism}

When defense is active, let
\[
 D_\tau(\theta;b)=
 \Phinorm\!\left(\sqrt\alpha(y_H-\theta)\right)
 -\Phinorm\!\left(\sqrt\alpha(y_L-\theta)\right)
\]
be its probability conditional on the state.  The marginal trader's posterior density is \(f_B(u)=\sqrt B\phi(\sqrt B u)\), centered at \(m(b)=m_S(T_\beta(b))\).  Define the intervention loss
\begin{equation}\label{eq:loss}
 L_\tau(b)=\int_{a(b)}^bD_\tau(\theta;b)
 f_B(\theta-m(b))\,d\theta,
\end{equation}
and set \(L_\tau(b)=0\) when defense is inactive.  The cutoff equation is
\begin{equation}\label{eq:residual}
 F_\tau(b)=R_\tau(b)-L_\tau(b)=0.
\end{equation}

The intervention loss admits a finite-CDF lower bound.  For an active candidate, put \(C=B+\alpha\), \(\kappa=\alpha B/(\alpha+B)\), and
\[
 \widehat m(y)=\frac{Bm(b)+\alpha y}{C}.
\]
Define
\[
 Q_b(y)=\Phinorm\!\left(\sqrt C\,[b-\widehat m(y)]\right)
 -\Phinorm\!\left(\sqrt C\,[a-\widehat m(y)]\right)
\]
and let
\[
 G_b(y)=\Phinorm\!\left(\sqrt\kappa\,[y-m(b)]\right)
\]
be the conditional CDF of the defender's signal given \(x=x^*(b)\).

\begin{lemma}\label{lem:lowerbound}
For every finite partition

\[
 y_L=s_0<s_1<\cdots<s_K=y_H,
\]
the intervention loss satisfies
\begin{equation}\label{eq:lowerbound}
 L_\tau(b)\geq \underline L_\tau(b;\mathbf s)
 :=\sum_{j=1}^K
 \min\{Q_b(s_{j-1}),Q_b(s_j)\}
 \,[G_b(s_j)-G_b(s_{j-1})].
\end{equation}
The bound increases weakly under partition refinement and converges to \(L_\tau(b)\) as the mesh tends to zero.
\end{lemma}

\begin{proof}
Conditional on the marginal small trader's signal, \(\theta\sim\Norm(m(b),1/B)\).  Since \(y=\theta+\eta\), \(y\sim\Norm(m(b),1/\kappa)\), and conditioning once more gives
\[
 \theta\mid y,x^*(b)\sim\Norm(\widehat m(y),1/C).
\]
Consequently, Fubini's theorem gives
\[
 L_\tau(b)
 =\Pr(a<\theta\leq b,\ y_L<y<y_H\mid x^*(b))
 =\int_{y_L}^{y_H}Q_b(y)\,dG_b(y).
\]
The function \(Q_b\) is symmetric and single-peaked as \(\widehat m(y)\) passes through \(\mu=(a+b)/2\): at distance \(t\geq0\) from \(\mu\), it equals
\[
 \Phinorm(\sqrt C[h-t])+\Phinorm(\sqrt C[h+t])-1,
\]
whose derivative is strictly negative for \(t>0\).  Its minimum on each partition cell is therefore attained at an endpoint.  Hence \(\underline L_\tau\) is the lower Darboux--Stieltjes sum of the continuous function \(Q_b\) with respect to the probability measure induced by \(G_b\). Refinement can only increase the sum.  If \(\omega_{Q_b}\) is the modulus of continuity and \(\|\mathbf s\|\) is the mesh, then
\[
 0\leq L_\tau(b)-\underline L_\tau(b;\mathbf s)
 \leq \omega_{Q_b}(\|\mathbf s\|)
 [G_b(y_H)-G_b(y_L)]\longrightarrow0.
\]
\end{proof}

The decomposition makes the feedback precise.  The benchmark margin \(R_\tau\) falls monotonically.  Intervention subtracts a nonnegative loss. At an activation boundary the defense band shrinks continuously to a point, so \(L_\tau\to0\).  Moreover, \(L_\tau(b)\to0\) as \(b\downarrow0\): the pivotal interval converges to \([-\lambda,0]\), while the marginal trader's posterior mean tends to \(-\infty\).  Thus the loss can be large only between the two ends of the left active-defense component.

\begin{proposition}\label{prop:hump}
Suppose \eqref{eq:diffuse} holds, \(c_L<\overline H\), and \(H_\tau(a(b_0),b_0)<c_L\).  Let \((0,\bar b)\) be the connected component of \(\{b:H_\tau(a(b),b)>c_L\}\) whose closure contains zero.  Then \(\bar b<b_0\) and \(L_\tau(0+)=L_\tau(\bar b)=0\).  Choose \(\widetilde b\in(0,\bar b)\) and a finite partition \(\mathbf s\) of its defense band.  The following finite-CDF condition is sufficient for two active-defense roots in addition to \(b_0\):
\begin{equation}\label{eq:overshoot}
 c_S>
 \Phinorm\!\left(\sqrt B\,[\widetilde b-m(\widetilde b)]\right)
 -\underline L_\tau(\widetilde b;\mathbf s).
\end{equation}
Hence \eqref{eq:residual} has at least three distinct roots.  If one selected root in each of \((0,\widetilde b)\) and \((\widetilde b,\bar b)\) satisfies the conditions of Lemma~\ref{lem:global}, those roots and \(b_0\) are symmetric monotone-cutoff Bayesian equilibria.
\end{proposition}

\begin{proof}
Because \(c_L<\overline H\), defense is active for all sufficiently small \(b\).  The strict inequality at \(b_0\) places a neighborhood of the benchmark root outside that active component, hence \(\bar b<b_0\).  By continuity at activation, \(L_\tau(\bar b)=0\), while \(R_\tau(\bar b)>0\).  At the other end, \(F_\tau(b)\to1-c_S>0\) as \(b\downarrow0\).  Lemma~\ref{lem:lowerbound} and condition \eqref{eq:overshoot} give \(L_\tau(\widetilde b)>R_\tau(\widetilde b)\), so \(F_\tau(\widetilde b)<0\).  The intermediate value theorem therefore supplies one active root on each side of \(\widetilde b\).  Since defense is inactive at \(b_0\), that benchmark root is a third.  Lemma~\ref{lem:global} converts each active marginal root into a global cutoff best response whenever its hypotheses hold.
\end{proof}

\begin{corollary}\label{cor:costwindow}
Suppose \eqref{eq:diffuse} holds, and define \(c_-=H_\tau(a(b_0),b_0)\) and \(c_+=\overline H\).  The coexistence route in Proposition~\ref{prop:hump}, with two active-defense roots and an inactive benchmark root, requires \(c_L\in(c_-,c_+)\).  Within this interval, \eqref{eq:overshoot} together with the global-incentive conditions is sufficient for at least three symmetric monotone-cutoff equilibria.
\end{corollary}

\begin{proof}
If \(c_L\geq c_+\), Proposition~\ref{prop:unique} makes defense inactive at every candidate.  If \(c_L\leq c_-\), the benchmark root is not strictly inactive.  For \(c_L\in(c_-,c_+)\), the remaining sufficiency statements are exactly Proposition~\ref{prop:hump} and Lemma~\ref{lem:global}.
\end{proof}

The interval is not by itself sufficient: the loss must still overshoot the benchmark margin.  Its role is to isolate the economic content of the construction.  Defense must be cheap enough to operate at pessimistic cutoffs but costly enough to disappear at the benchmark cutoff.

The argument concerns the left active component.  It does not assert that defense remains inactive at every larger cutoff; an active component may reappear near \(b=1\).  The certified example below excludes additional roots there by a separate right-tail bound.

It remains to check full incentives.  With active defense, set \(n(\theta)=1-D_\tau(\theta)\).  It is symmetric and U-shaped about \(\gamma\). Define
\[
 \ell=\gamma-a,\qquad r=b-\gamma,\qquad
 n_a=n(a),\quad n_b=n(b),\quad n_0=n(\gamma).
\]
Equation \eqref{eq:center} gives
\begin{equation*}\tag{O}\label{eq:ordering}
 \gamma-\mu=\frac{\tau(\mu-m_0)}{\alpha},
 \qquad
 r-\ell=\frac{2\tau(m_0-\mu)}{\alpha}.
\end{equation*}

\begin{lemma}\label{lem:global}
Fix an active-defense candidate.  Suppose \(a<\gamma<b\), \(0<\ell<r\), and
\begin{align}
 n_b-n(\gamma+\ell)&<1-n_a, \tag{L}\label{eq:L}\\
 n_b-n_0&<(1-n_a)e^{-2Bh^2}+n_be^{-Br^2/2}. \tag{GM}\label{eq:GM}
\end{align}
Then a small trader's success probability is strictly decreasing in her signal.  Consequently, any zero of \(F_\tau\) satisfying these conditions is a global cutoff best response.
\end{lemma}

\begin{proof}
Let \(S(x;x^*)\) be the posterior probability that an attack succeeds and write \(m=m_S(x)\).  Then
\begin{equation}\label{eq:success}
 S(x;x^*)=\Phinorm\!\left(\sqrt B(a-m)\right)
 +\int_a^b n(\theta)f_B(\theta-m)\,d\theta .
\end{equation}
Differentiation and integration by parts give
\begin{equation*}\tag{D}\label{eq:derivative}
 \frac{dS}{dm}=-(1-n_a)f_B(a-m)-n_bf_B(b-m)
 +\int_a^b n'(\theta)f_B(\theta-m)\,d\theta .
\end{equation*}
The first two terms are the downward jumps in state-contingent success at \(a\) and \(b\).

If \(m\leq\gamma\), pair \(\gamma-z\) with \(\gamma+z\) for \(z\in[0,\ell]\).  Antisymmetry of \(n'\) and the Gaussian density ordering make every paired contribution nonpositive.  The only remaining positive variation lies on \([\gamma+\ell,b]\).  Its contribution is at most \([n_b-n(\gamma+\ell)]f_B(a-m)\), which the lower jump dominates by \eqref{eq:L}.

If \(\gamma\leq m\leq b\), discard the negative variation on \([a,\gamma]\).  The positive variation is at most \((n_b-n_0)f_B(0)\).  Since \(|a-m|\leq2h\) and \(|b-m|\leq r\), condition \eqref{eq:GM} makes the two downward jumps larger.  Finally, if \(m\geq b\), positive interior variation is at most \((n_b-n_0)f_B(b-m)\).  After the upper jump, \eqref{eq:derivative} is at most \(-(1-n_a)f_B(a-m)-n_0f_B(b-m)<0\).  Thus \(S\) is strictly decreasing.
\end{proof}

\section{A certified local exact count}

Consider the proper prior \((\tau,m_0)=(0.001,0)\), so \(\theta\sim\Norm(0,1000)\), and the parameter vector
\begin{equation}\label{eq:parameters}
 (\alpha,\beta,\lambda,c_L,c_S)=(20,8,0.8,0.71,0.919).
\end{equation}
For these values, \(c_-\approx0.675953\) and \(c_+\approx0.926369\), placing the defense cost inside the intermediate region identified above.  The left active component ends at \(\bar b\approx0.048960625\), whereas the inactive benchmark root is \(b_0\approx0.080991182\).  At \(\widetilde b=0.020\), let \(\underline L_{\tau,2}\) denote the bound generated by the partition \((y_L,\gamma,y_H)\).  The certificate gives
\[
 \underline L_{\tau,2}>0.0758838764
 \quad\text{and}\quad
 R_\tau(\widetilde b)<0.0609941273.
\]
Thus the finite-CDF overshoot margin exceeds \(0.0148897492\), without numerical quadrature.  Direct integration gives \(L_\tau(\widetilde b)\approx0.08236111\) as a corroborating point calculation.

\begin{figure}[t]
\centering
\includegraphics[width=\textwidth]{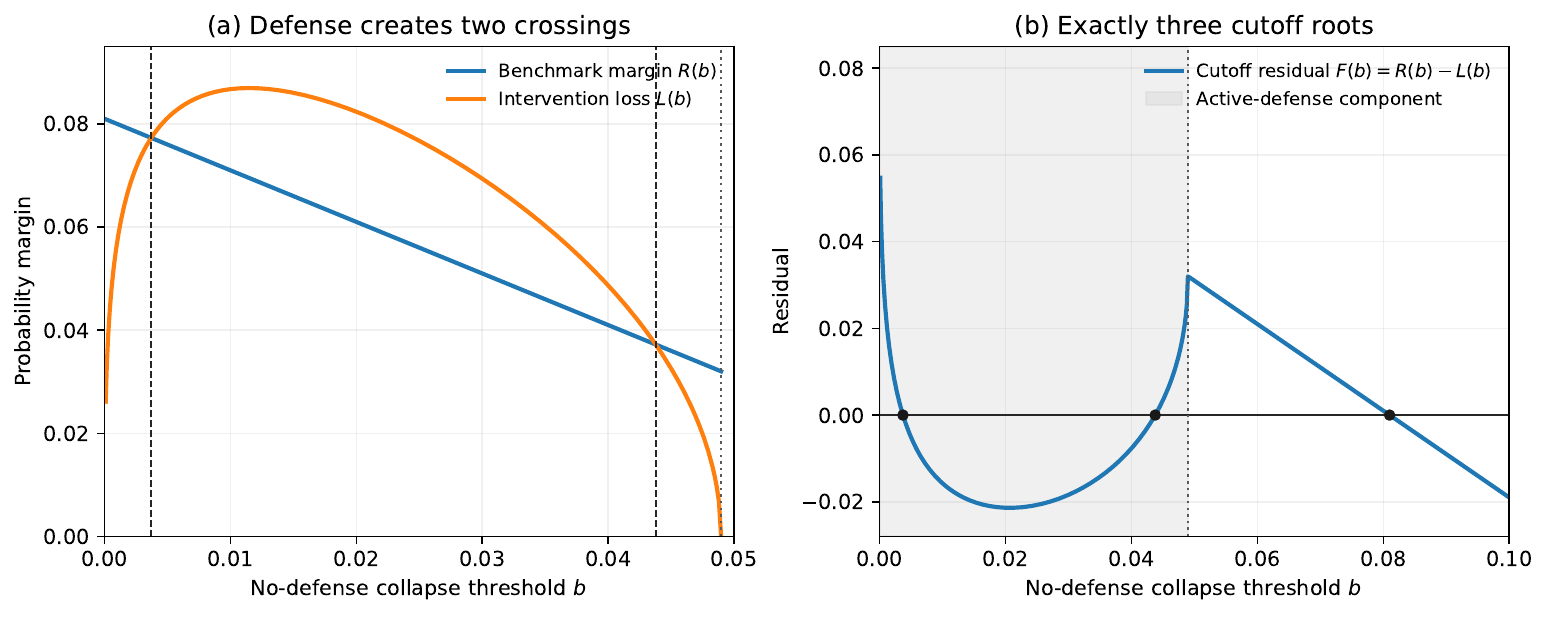}
\caption{The defense-loss mechanism at \eqref{eq:parameters}. Curves are high-precision point evaluations for exposition; the three roots and their exhaustion are established by the outward-rounded interval certificates. In panel (a), the dashed lines mark the two active roots and the dotted line the end of the left active-defense component.}
\label{fig:mechanism}
\end{figure}

\begin{proposition}\label{prop:main}
At \eqref{eq:parameters}, the game has exactly three distinct symmetric equilibrium cutoff configurations, corresponding to pure-strategy Bayesian equilibria in monotone small-trader cutoffs under the null-type convention. Two feature intermediate-signal defense and one features no defense.  Holding \((\tau,m_0)=(0.001,0)\) fixed, this conclusion persists on a nonempty open neighborhood of the five-dimensional vector in \eqref{eq:parameters}.
\end{proposition}

\begin{table}[t]
\centering
\caption{Outward-rounded endpoint signs and uniform bracket margins}
\label{tab:certificate}
\small
\begin{tabular}{lcccc}
\toprule
\(b\)-bracket & \(F_\tau\) signs & \(H_\tau-c_L\) & (L) margin & (GM) margin\\
\midrule
\([0.0035,0.0040]\) & \((+,-)\) & \(>0.123\) & \(>0.240\) & \(>0.220\)\\
\([0.0435,0.0440]\) & \((-,+)\) & \(>0.0064\) & \(>0.061\) & \(>0.680\)\\
\([0.0805,0.0815]\) & \((+,-)\) & \(<-0.032\) & --- & ---\\
\bottomrule
\end{tabular}
\end{table}

\begin{proof}
Ball arithmetic gives the six endpoint signs in Table~\ref{tab:certificate}.  Continuity supplies a root in each disjoint bracket.  On the first two brackets, defense is uniformly active and the certificate verifies \(a<\gamma<b\), \(0<\ell<r\), \eqref{eq:L}, and \eqref{eq:GM} uniformly.  Lemma~\ref{lem:global} therefore turns both roots into equilibria.  On the third bracket, defense is uniformly inactive and the success probability is strictly decreasing.  Exact aggregation and the defender's interval best response complete the equilibrium construction.

For exactness, a second certificate exhausts \(b\in(0,1)\) with seven pieces (reported with their interval margins in the certificate ledger below): a positive left tail; a strictly decreasing segment containing the first root; a strictly negative trough; a strictly increasing segment containing the second root; a positive activation-transition segment; a strictly decreasing inactive segment containing the third root; and a negative right tail.  Hence \(F_\tau\) has no other zeros.  Every equilibrium in the stated class induces a unique \(b\in(0,1)\) and hence a zero of \(F_\tau\).  Because \(T_\beta\) is a bijection from \((0,1)\) onto \(\R\), the parameterization omits no cutoff.  Therefore there are no further equilibria in the stated class.

All tail, sign, derivative, regime, geometry, and incentive inequalities in the fixed finite cover are strict and continuous in the five primitives in \eqref{eq:parameters}, holding \((\tau,m_0)=(0.001,0)\) fixed.  The same cover therefore remains valid after sufficiently small perturbations, proving the open-neighborhood statement.
\end{proof}

The isolated roots are approximately
\[
 b_1=0.0037124555,
 \qquad b_2=0.0437670125,
 \qquad b_3=0.0809911822.
\]
The first two are not replicas of the usual inactive root.  They lie inside the left defense component, on opposite sides of the intervention-loss hump. The third is the perturbed proper-prior analogue of \(1-c_S\).

\begin{table}[t]
\centering
\caption{Approximate equilibrium objects at the certified vector}
\label{tab:outcomes}
\small
\begin{tabular}{lcccc}
\toprule
 & \(b\) & \(x^*\) & pivotal states \([a,b]\) & defense signals \([y_L,y_H]\)\\
\midrule
1 & 0.0037125 & -0.942807 & [-0.619649,\ 0.003712] & [-0.486569,\ -0.129398]\\
2 & 0.0437670 & -0.560297 & [-0.436688,\ 0.043767] & [-0.235606,\ -0.157334]\\
3 & 0.0809912 & -0.413430 & [-0.360029,\ 0.080991] & inactive\\
\bottomrule
\end{tabular}
\end{table}

The first active equilibrium combines the lowest attack cutoff with the widest defense-signal band.  At the second, small traders attack at higher signals while the defense band narrows.  At the third, the defense band disappears and the attack cutoff rises again.  The two active crossings thus represent distinct self-consistent assessments of protection---more expected protection with less aggressive small traders, and a narrower defense response with more aggressive small traders---rather than duplicate numerical roots.

\section{Interpretation and scope}

The result depends on three distinct forces.  Strategic complementarity among small traders creates the usual cutoff logic.  The defender's finite capacity makes intervention pivotal only on a state interval.  His private signal then turns that state interval into an intermediate-signal defense band.  Because his choice is contemporaneous and hidden, the band changes payoffs but not beliefs.  It is this payoff effect---not policy signalling and not an exogenous perturbation of higher-order beliefs---that bends the residual.

The cost comparison is equally important.  If defense is too expensive, Proposition~\ref{prop:unique} applies and the benchmark survives.  If defense is active at pessimistic cutoffs but inactive at \(b_0\), the intervention loss can generate a pair of additional equilibria.  This gives a simple economic interpretation of the three roots: limited and selectively deployed defense need not stabilize expectations; it can create two self-consistent views about how much protection the market will face.

The currency-attack labels are not essential to this logic.  The analytical object is a decreasing benchmark residual minus an intervention loss supported on an endogenous activation component.  Whenever hidden action is valuable only for intermediate fundamentals, the loss can vanish at the component's boundaries and peak inside it.  Proposition~\ref{prop:hump} states the resulting overshoot test explicitly; the Gaussian structure makes the test finite and computable.  Applications outside currency crises would require their own verification of the global-incentive and activation conditions.

The word ``local'' is substantive.  The exact-count neighborhood in Proposition~\ref{prop:main} is topological and unquantified.  To report a numerical scale without conflating existence and exactness, a separate certificate holds \((\tau,m_0,\alpha,\beta,\lambda)\) fixed and establishes at least three qualified equilibria throughout the independent rectangle
\[
 (c_L,c_S)\in[0.70,0.72]\times[0.909,0.929].
\]
This rectangle does not carry an exact-count claim.

The paper does not claim a general characterization of every Bayesian equilibrium.  Proposition~\ref{prop:main} is complete only within symmetric monotone small-trader cutoffs.  Asymmetric and non-cutoff equilibria remain outside its scope.  The parameter vector is constructed rather than calibrated, and the result is not evidence about the empirical frequency of currency crises.  Its purpose is sharper: to identify a transparent intervention technology under which the familiar global-game selection argument changes, to state the condition responsible for that change, and to verify one locally persistent exact count without relying on a plot or a floating-point grid.

\appendix
\section{Exhaustive certificate ledger}\label{app:ledger}

Table~\ref{tab:ledger} records the complete cover used to exclude every other zero of \(F_\tau\).  On active cells, the program differentiates the threshold equations and the defender's band equation implicitly and then applies Leibniz' rule to \eqref{eq:loss}.  Every displayed derivative or residual bound is an outward-rounded Arb enclosure.  The two active root brackets also satisfy Lemma~\ref{lem:global} uniformly, as reported in Table~\ref{tab:certificate}.

\begin{table}[!htbp]
\centering
\caption{Seven-piece exhaustive cover of \(b\in(0,1)\)}
\label{tab:ledger}
\small
\begin{tabular}{>{\raggedright\arraybackslash}p{0.20\textwidth}
                >{\raggedright\arraybackslash}p{0.39\textwidth}
                >{\raggedright\arraybackslash}p{0.29\textwidth}}
\toprule
Cutoff segment & Certified fact & Outward-rounded margin\\
\midrule
\((0,0.0001]\) & \(F_\tau>0\) from success below \(a\) & at least \(0.0270838\)\\
\([0.0001,0.018]\) & active and \(F_\tau'<0\); contains \(b_1\) & 204 cells; \(-F_\tau'>0.0818207\)\\
\([0.018,0.024]\) & \(F_\tau<0\) & 120 cells; \(-F_\tau>0.0209890\)\\
\([0.024,0.048]\) & active and \(F_\tau'>0\); contains \(b_2\) & 240 cells; \(F_\tau'>0.143494\)\\
\([0.048,0.050]\) & \(F_\tau>0\), including the regime transition & 200 cells; \(F_\tau>0.0137648\)\\
\([0.050,0.090]\) & inactive and benchmark residual decreasing; contains \(b_3\) & 40 cells; \(c_L-H_\tau>0.0003563\)\\
\([0.090,1)\) & \(F_\tau\leq R_\tau<0\) & \(R_\tau(0.090)<-0.0090084\)\\
\bottomrule
\end{tabular}
\end{table}

The sign-changing brackets are \([0.0035,0.0040]\), \([0.0435,0.0440]\), and \([0.0805,0.0815]\). The computations use Python 3.12.13, \texttt{python-flint} 0.8.0, FLINT 3.3.1, and 70 decimal digits of working precision. Source, pinned requirements, complete transcripts, and a checksum manifest are included in the \texttt{anc/} directory of the arXiv source package.

\section{Ancillary certificates and reproduction}\label{app:anc}

The \texttt{anc/} directory contains the complete machine-readable supplement. Each proof program has a matching archived transcript, and \path{anc/SHA256SUMS.txt} authenticates every ancillary file.

\paragraph{Core equilibrium certificate.} \path{yosifov_newvector_core_certificate.py} verifies the three root brackets, defender regimes, geometry, and uniform global-incentive conditions at the certified vector.

\paragraph{Exact-count certificate.} \path{yosifov_newvector_exactly_three_certificate.py} combines analytic tails with the seven-piece outward-rounded cover in Table~\ref{tab:ledger} to exclude every other zero on \(b\in(0,1)\).

\paragraph{Cost-rectangle certificate.} \path{yosifov_newvector_cost_rectangle_certificate.py} proves at least three qualified equilibria throughout \([0.70,0.72]\times[0.909,0.929]\) in \((c_L,c_S)\), holding the other primitives fixed. It does not assert exactness throughout that rectangle.

\paragraph{Overshoot and reported-value certificates.} \path{yosifov_newvector_overshoot_bound_certificate.py} proves the two-bin finite-CDF overshoot at \(b=0.020\) without numerical quadrature. \path{yosifov_newvector_reported_values_certificate.py} encloses the cost thresholds, activation boundary, direct-loss corroboration, and the approximate equilibrium objects reported in the text.

\paragraph{Reproduction.} After verifying the checksum manifest, install the pinned dependency in \path{anc/requirements-yosifov-newvector-certificate.txt} and run each certificate from the \texttt{anc/} directory. The archived transcripts were generated at 70 decimal digits; a successful rerun ends with the corresponding \texttt{PASSED} message and reproduces its archived transcript byte for byte.

\end{document}